\documentclass[a4paper,UKenglish,cleveref,autoref, thm-restate]{lipics-v2021}

\usepackage[noline,noend,linesnumbered,algo2e,ruled]{algorithm2e}
\usepackage{todonotes}

\newcommand{\cO}{\mathcal{O}}
\newcommand{\cOtilde}{\tilde{\cO}}
\def\dd{\mathinner{.\,.}}

\newcommand{\per}{\textsf{per}}

\title{Substring Complexity in Sublinear Space} 

\author{Giulia Bernardini}{University of Trieste, Italy}{giulia.bernardini@units.it}{https://orcid.org/0000-0001-6647-088X}{MUR - FSE REACT EU - PON R\&I 2014-2020.}

\author{Gabriele Fici}{Dipartimento di Matematica e Informatica, University of Palermo, Italy}{gabriele.fici@unipa.it}{https://orcid.org/0000-0002-3536-327X}{Projects MUR PRIN~2017 ADASCOML -- 2017K7XPAN and MUR PRIN~2022 APML -- 20229BCXNW.}

\author{Paweł Gawrychowski}{Institute of Computer Science, University of Wroc\l{}aw, Poland}{gawry@cs.uni.wroc.pl}{https://orcid.org/0000-0002-6993-5440}{}

\author{Solon P. Pissis}{CWI, Amsterdam, The Netherlands \and Vrije Universiteit, Amsterdam, The Netherlands }{solon.pissis@cwi.nl}{https://orcid.org/0000-0002-1445-1932}{Supported by the PANGAIA (No 872539) and ALPACA (No 956229) projects.}

\authorrunning{G. Bernardini, G. Fici, P. Gawrychowski, S. P. Pissis} 

\Copyright{Giulia Bernardini, Gabriele Fici, Pawel Gawrychowski, and Solon P. Pissis} 



\category{} 

\supplement{}

\nolinenumbers 

\hideLIPIcs  

\EventEditors{Satoru Iwata and Naonori Kakimura}
\EventNoEds{2}
\EventLongTitle{34th International Symposium on Algorithms and Computation (ISAAC 2023)}
\EventShortTitle{ISAAC 2023}
\EventAcronym{ISAAC}
\EventYear{2023}
\EventDate{December 3-6, 2023}
\EventLocation{Kyoto, Japan}
\EventLogo{}
\SeriesVolume{283}
\ArticleNo{11}

\begin{document}

\maketitle

\begin{abstract}
Shannon's entropy is a definitive lower bound for statistical compression. Unfortunately, no such clear measure exists for the compressibility of repetitive strings. Thus, ad hoc measures are employed to estimate the repetitiveness of strings, e.g., the size $z$ of the Lempel–Ziv parse or the number $r$ of equal-letter runs of the Burrows-Wheeler transform. A more recent one is the size $\gamma$ of a smallest string attractor. Let $T$ be a string of length $n$.
A string attractor of $T$ is a set of positions of $T$ capturing the occurrences of all the substrings of $T$. Unfortunately, Kempa and Prezza~[STOC 2018] showed that computing $\gamma$ is NP-hard. Kociumaka et al.~[LATIN 2020] considered a new measure of compressibility that is based on the function $S_T(k)$ counting the number of distinct substrings of length $k$ of $T$, also known as the {\em substring complexity} of $T$. This new measure is defined as $\delta= \sup\{S_T(k)/k, k\geq 1\}$ and lower bounds all the relevant ad hoc measures previously considered. In particular, $\delta\leq \gamma$ always holds and $\delta$ can be computed in $\cO(n)$ time using $\Theta(n)$ working space. Kociumaka et al.~showed that one can construct an $\cO(\delta \log \frac{n}{\delta})$-sized representation of $T$ supporting efficient direct access and efficient pattern matching queries on $T$. Given that for highly compressible strings, $\delta$ is significantly smaller than $n$, it is natural to pose the following question: 
\begin{center}
\emph{Can we compute $\delta$ efficiently using sublinear working space?} 
\end{center}

It is straightforward to show that in the comparison model, any algorithm computing $\delta$ using $\cO(b)$ space requires $\Omega(n^{2-o(1)}/b)$ time through a reduction from the element distinctness problem [Yao, SIAM J. Comput. 1994]. We thus wanted to investigate whether we can indeed match this lower bound. We address this algorithmic challenge by showing the following bounds to compute $\delta$: 

\begin{itemize}
    \item $\cO(\frac{n^3\log b}{b^2})$ time using $\cO(b)$ space, for any $b\in[1,n]$, in the comparison model.
    \item $\cOtilde(n^2/b)$\footnote{The $\tilde{\cO}(f)$ notation denotes $\cO(f\cdot \text{polylog}(n))$.} time using $\cOtilde(b)$ space, for any $b\in[\sqrt{n},n]$, in the word RAM model. This gives an $\cOtilde(n^{1+\epsilon})$-time and $\cOtilde(n^{1-\epsilon})$-space algorithm to compute $\delta$, for any $0<\epsilon \leq 1/2$.
\end{itemize}

Let us remark that our algorithms compute $S_T(k)$, for all $k$, within the same complexities.
\end{abstract}

\section{Introduction}\label{sec:intro}

We are currently witnessing our world drowning in data. These datasets are generated by a large gamut of applications: databases, web applications, genome sequencing projects, scientific computations, sensors, e-mail, entertainment, and others. The biggest challenge is thus to develop theoretical and practical methods for processing datasets efficiently.

Compressed data representations that can be {\em directly} used in compressed form have a central role in this challenge~\cite{DBLP:books/daglib/0038982}. Indeed, much of the currently fastest-growing data is highly repetitive; this, in turn, enables space reductions of orders of magnitude~\cite{DBLP:journals/jacm/GagieNP20}. Prominent examples of such data include genome, versioned text, and software repositories collections. A common characteristic is that each element in a collection is very similar to every other.

Since a significant amount of this data is sequential, a considerable amount of algorithmic research has been devoted to text indexes over the past decades~\cite{DBLP:conf/focs/Weiner73,DBLP:journals/siamcomp/ManberM93,DBLP:conf/focs/Farach97,DBLP:journals/jacm/KarkkainenSB06,DBLP:journals/jacm/FerraginaM05,DBLP:journals/siamcomp/GrossiV05,DBLP:journals/siamcomp/HonSS09,DBLP:conf/stoc/Belazzougui14,DBLP:journals/algorithmica/0001KL15,DBLP:conf/soda/MunroNN17,DBLP:conf/stoc/KempaK19,DBLP:journals/jacm/GagieNP20,DBLP:conf/soda/KempaK23}. String processing applications (see~\cite{DBLP:books/cu/Gusfield1997,DBLP:journals/cacm/ApostolicoCFGM16} for reviews) require fast access to the substrings of the input string. These applications rely on such text indexes, which arrange the string suffixes lexicographically in an ordered tree~\cite{DBLP:conf/focs/Weiner73} or an ordered array~\cite{DBLP:journals/siamcomp/ManberM93}.

This significant  amount of research has resulted in compressed text indexes
that support fast pattern searching in
space close to the statistical entropy of the text collection.
The problem, however, is that this kind of entropy is unable to capture repetitiveness~\cite{DBLP:journals/tcs/KreftN13,DBLP:journals/jcb/MakinenNSV10}. To achieve orders-of-magnitude space reductions, one thus needs to resort to other compression methods, such as Lempel-Ziv (LZ)~\cite{LZ77}, grammar compression~\cite{grammar} or run-length compressed Burrows-Wheeler transform (BWT)~\cite{DBLP:journals/jacm/GagieNP20}, to name a few; see~\cite{DBLP:journals/jacm/GagieNP20} for a review.

Unlike Shannon's entropy, which is a definitive lower bound for statistical compression, no such clear measure exists for the compressibility of repetitive texts.
Other than Kolmogorov’s complexity~\cite{Kolmogorov}, which is not computable, repetitiveness is measured in ad hoc terms, based on what the compressors may achieve. Such measures on a string $T$ include: the number $z$ of phrases produced by the LZ parsing of $T$; the size $g$ of the smallest grammar generating $T$; and the number $r$ of maximal equal-letter runs in the BWT of $T$. See~\cite{DBLP:journals/csur/Navarro21a} for a survey.

An improvement is the recent introduction of the {\em string attractor}~\cite{DBLP:conf/stoc/KempaP18} notion. Let $T$ be a string of length $n$.
An attractor $\Gamma$ is a set of positions over $[1,n]$ such that any substring of $T$ has an occurrence covering a position in $\Gamma$. The size $\gamma$ of a smallest attractor asymptotically lower bounds all the repetitiveness measures listed above (and others; see~\cite{KNPlatin20}). Unfortunately, using indexes based on $\gamma$ comes also with some challenges. Other than computing $\gamma$ is NP-hard~\cite{DBLP:conf/stoc/KempaP18}, it is unclear if $\gamma$ is the definitive measure of repetitiveness: we do not know whether one can always represent $T$ in $\cO(\gamma)$ space (machine words). This motivated Christiansen et al.~\cite{DBLP:journals/corr/abs-1811-12779} to consider a new measure $\delta$ of compressibility, initially introduced in the area of string compression by Raskhodnikova et al.~\cite{DBLP:journals/algorithmica/RaskhodnikovaRRS13}, and for which $\delta \leq \gamma$ \emph{always holds}~\cite{DBLP:journals/corr/abs-1811-12779}.

\begin{definition}[\cite{DBLP:journals/corr/abs-1811-12779}]
Let $T$ be a string and $S_T(k)$ its {\em substring complexity}: the function counting the number of distinct substrings of length $k$ of $T$. The {\em normalized substring complexity} of $T$ is the function $S_T(k)/k$ and we set $\delta= \sup\{S_T(k)/k, k\geq 1\}$ its supremum.
\end{definition}

Christiansen et al.~also showed that $\delta$ can be computed in $\cO(n)$ time using $\Theta(n)$ working space. Kociumaka et al.~\cite{KNPlatin20,DBLP:journals/tit/KociumakaNP23} showed that $\delta$ can also be strictly smaller than $\gamma$ by up to a logarithmic factor: for any $n$ and any $\delta$, there are strings with $\gamma=\Omega(\delta \log \frac{n}{\delta})$. Moreover, Kociumaka et al.~developed a representation of $T$ of size $\cO(\delta \log \frac{n}{\delta})$, which is worst-case optimal in terms of $\delta$ and allows for accessing any $T[i]$ in time $\cO(\log \frac{n}{\delta})$ and for finding all \textit{occ} occurrences of any pattern $P[1\dd m]$ in $T$ in near-optimal time $\cO(m\log n+ \textit{occ}\log^{\epsilon}n)$, for any constant $\epsilon>0$ (see also~\cite{DBLP:conf/latin/KociumakaNO22} and~\cite{DBLP:journals/corr/abs-2308-03635} for further improvements). Since for highly compressible strings, $\delta$ is significantly smaller than $n$, we pose the following basic question: 
\begin{center}
\textit{Can we compute $\delta$ efficiently using sublinear working space?}
\end{center}

The question on computing $\delta$ in \emph{bounded space} arises naturally: it extends a large body of work on problems on strings, which admit a straightforward solution if we have the space to construct and store the suffix tree~\cite{DBLP:conf/focs/Weiner73}; but as this is often not the case, one needs to overcome the space challenge by investigating space-time trade-offs for these problems.

\subparagraph*{Related Work.}~The standard approach for showing space-time trade-off lower bounds for problems answered in polynomial time has been to analyze their complexity on {\em (multi-way) branching programs}. In this model, the input is stored in read-only memory, the output in write-only memory, and neither is counted towards the space used by any algorithm. This model is powerful enough to simulate both Turing machines and standard RAM models that are unit-cost with respect to time and log-cost with respect to space. It was introduced by Borodin and Cook, who used it to prove that any multi-way branching program requires a time-space product of  $\Omega(n^2/\log n)$ to sort $n$ integers in the range $[1,n^2]$~\cite{DBLP:journals/siamcomp/BorodinC82,DBLP:journals/siamcomp/Beame91}. Unfortunately, the techniques in~\cite{DBLP:journals/siamcomp/BorodinC82} yield only trivial bounds for problems with single outputs.

String algorithms that use sublinear space have been extensively studied over the past decades~\cite{DBLP:journals/jcss/GalilS83,DBLP:journals/jacm/CrochemoreP91,DBLP:conf/focs/PoratP09,DBLP:journals/tcs/BreslauerGM13,DBLP:conf/cpm/StarikovskayaV13,DBLP:journals/talg/BreslauerG14,DBLP:conf/esa/KociumakaSV14,DBLP:conf/esa/CliffordFPSS15,DBLP:conf/esa/0001GGK15,DBLP:conf/soda/CliffordFPSS16,DBLP:conf/icalp/CliffordS16,DBLP:conf/esa/GolanP17,DBLP:journals/algorithmica/GolanKP19,DBLP:conf/icalp/GolanKP18,DBLP:conf/cpm/GolanKKP20,DBLP:conf/soda/CliffordKP19,DBLP:conf/cpm/GawrychowskiS19,DBLP:conf/dcc/AyadBFHP19,DBLP:conf/cpm/Nun0KK20,DBLP:conf/stoc/ChanGKKP20,DBLP:journals/iandc/RadoszewskiS20,algorithmica2020}.
The perhaps most relevant problem to our work is the classic longest common substring of two strings. Formally, given two strings $X$ and $Y$ of total length $n$, the \emph{longest common substring} (LCS) problem consists in computing a longest string occurring as a substring of both $X$ and $Y$. The LCS problem was conjectured by Knuth to require $\Omega(n \log n)$ time. This conjecture was disproved by Weiner who, in his seminal paper on suffix tree construction~\cite{DBLP:conf/focs/Weiner73}, showed how to solve the LCS problem in $\cO(n)$ time for constant-sized alphabets. Farach showed that the same problem can be solved in the optimal $\cO(n)$ time for polynomially-sized integer alphabets~\cite{DBLP:conf/focs/Farach97}. 
A straightforward space-time trade-off lower bound of $\Omega(b)$ space and $\Omega(n^2/b)$ time for the LCS problem can be derived from the problem of checking whether the length of an LCS is $0$; i.e., deciding if $X$ and $Y$ have a common letter or not. Thus, in some sense, the LCS problem can be seen as a generalization of the {\em element distinctness} problem: given $n$ elements over a domain $D$,  decide whether all $n$ elements are distinct.

On the upper bound side, Starikovskaya and Vildh{\o}j showed that for any $b\in [n^{2/3},n]$, the LCS problem can be solved in $\cOtilde(n^2/b)$ time and $\cO(b)$ space~\cite{DBLP:conf/cpm/StarikovskayaV13}.
In~\cite{DBLP:conf/esa/KociumakaSV14}, Kociumaka et al.~gave an $\cO(n^2/b)$-time algorithm to find an LCS for any $b\in[1,n]$, and also provided a lower bound, which states that any deterministic multi-way branching program that uses $b\leq \frac{n}{\log n}$ space must take $\Omega(n\sqrt{\log(n/(b\log n))/\log \log(n/(b\log n))})$ time. This lower bound implies that the classic $\cO(n)$-time solution for the LCS problem~\cite{DBLP:conf/focs/Weiner73,DBLP:conf/focs/Farach97} is optimal in the sense that we cannot hope for an $\cO(n)$-time algorithm using $o(n/\log n)$ space. Unfortunately, we do not know if the $\cO(b)$-space and $\cO(n^2/b)$-time trade-off is generally the best possible for the LCS problem.
For the easier element distinctness problem, Beame et al.~\cite{DBLP:conf/focs/BeameCM13} showed a randomized multiway branching program using
$\cOtilde(n^{3/2}/\sqrt{b})$-time and $\cO(b)$ space.

It is thus a big open question to answer whether the LCS problem can be solved asymptotically faster than $\cO(n^2/b)$ using $\cO(b)$ space.
Towards this direction, Ben-Nun et al.~exploited the intuition suggesting that an LCS of $X$ and $Y$ can be computed more efficiently when its length $L$ is large~\cite{DBLP:conf/cpm/Nun0KK20} (see also~\cite{DBLP:conf/cpm/Charalampopoulos18}). The authors showed an algorithm which runs in $\cOtilde(\frac{n^2}{L\cdot b} + n)$ time, for any $b\in[1,n]$, using $\cO(b)$ space.
Still, a straightforward lower bound for the aforementioned problem is in $\Omega(\frac{n^2}{L^2\cdot b} + n)$ time when $\cO(b)$ space is used; it seems that further insight is required to match this space-time trade-off lower bound.

\subparagraph*{Our Results and Techniques.}~Our goal is to efficiently compute $\delta$ using $\cO(b)$ space. As a preliminary step towards this algorithmic challenge, we show the following theorem.

\begin{restatable}{theorem}{theoremncube}\label{the:n3}
    Given a string $T$ of length $n$, we can compute $\delta= \sup\{S_T(k)/k, k\geq 1\}$
    in $\cO(\frac{n^3\log b}{b^2})$ time using $\cO(b)$ space, for any $b\in[1,n]$, in the comparison model.
\end{restatable}

It is straightforward to show that any comparison-based branching program to compute $\delta$ using $\cO(b)$ space requires $\Omega(n^{2-o(1)}/b)$ time through a reduction from the element distinctness problem~\cite{DBLP:journals/siamcomp/Yao94}. By Yao's lemma, this lower bound also applies to randomized branching programs~\cite{DBLP:journals/siamcomp/Yao94}.
This suggests that a natural intermediate step towards fully understanding the computation complexity of computing $\delta$ in small space should be designing an $\cOtilde(n^{2}/b)$-time algorithm using $\cO(b)$ space (not necessarily in the comparison model). 

The natural approach for computing $\delta$ is
through computing all values of $S_T(k)$. 
In particular, this is the idea behind the straightforward $\cO(n)$-time computation of $\delta$ using $\cO(n)$ space~\cite{DBLP:journals/corr/abs-1811-12779}.
It is unclear to us if a more direct approach exists (see also Section~\ref{sec:combinatorial} for a combinatorial analysis on the behaviour of $\delta$). 
Under this plausible assumption, we stress that 
computing $S_T(k)$, for all $k$ one-by-one, is a more general problem than computing the length $L$ of an LCS of $X$ and $Y$,
as an algorithm computing $S_T(k)$ can be used to compute $L$ within the same complexities. 
This follows by the following argument:
we compute $S_X(k)$, $S_Y(k)$, and $S_{X\#Y}(k)$ (where $\#$ is a special letter that does not occur in $X$ or in $Y$)
in parallel, and set $L$ equal to the largest $k$ such that $S_X(k)+S_Y(k)>S_{X\#Y}(k)-k$.
As the best-known time upper bound for the very basic question of computing LCS in $\cO(b)$ space
remains to be $\cO(n^{2}/b)$, this further motivates the algorithmic challenge of designing an algorithm with such bounds for computing $\delta$. We address it by proving the following theorem.

\begin{restatable}{theorem}{theoremnsquare}\label{the:n2}
    Given a string $T$ of length $n$, we can compute $\delta= \sup\{S_T(k)/k, k\geq 1\}$
    in $\cOtilde(n^2/b)$ time using $\cOtilde(b)$ space, for any $b\in[\sqrt{n},n]$, in the word RAM model. 
\end{restatable}

Our algorithms compute $S_T(k)$, for all $k$, within the same complexities.
To arrive at the $\cO(\frac{n^3\log b}{b^2})$-time bound, we
split the computation of the values $S_T(k)$ in $n/b$ phases: in each phase, 
we restrict to substrings whose length is in a range of size $b$. 
In turn, in each phase, we process the substrings that start within a range of $b$ positions of $T$ at a time, from left to right. 
With this scheme, we process in $\cO(n\log b)$ time each block of $n/b$ positions of $T$ in each of the $n/b$ phases, resulting in $\cO(\frac{n^3\log b}{b^2})$ time using $\cO(b)$ space.
For large enough $b$, we can process all the substrings of a single phase at once, 
saving a factor of $n/b$. We show in fact  that a representation of all the occurrences of all the substrings of a phase can be packed in $\cOtilde(b)$ space if $b$ is large enough, and process them in different ways depending on their period, following a scheme similar to~\cite{DBLP:conf/icalp/0001GPPR19}; we  also adapt a method used in~\cite{bernardini2021elasticdegenerate} to select a small set of anchors (length-$b$ substrings), so that each fragment of $T$ contains at least one anchor but their total number of occurrences in $T$ is  bounded.
Note that Theorem~\ref{the:n2} implies an $\cOtilde(n^{1+\epsilon})$-time and $\cOtilde(n^{1-\epsilon})$-space algorithm to compute $\delta$, for any $0<\epsilon \leq 1/2$.

\subparagraph*{Paper Organization.}~Section~\ref{sec:prel} introduces the basic definitions and notation we use and the space-time trade-off lower bound for computing $\delta$. In Section~\ref{sec:simple}, we present a simple $\cO(n^3/b)$-time and $\cO(b)$-space algorithm, for any $b\in[1,n]$. This algorithm is refined to run in $\cO(\frac{n^3\log b}{b^2})$ time using $\cO(b)$ space, for any $b\in[1,n]$, in Section~\ref{sec:improve}. Our main result, the $\cOtilde(n^2/b)$-time and $\cOtilde(b)$-space algorithm, for any $b\in[\sqrt{n},n]$, is presented in Section~\ref{sec:main}. In Section~\ref{sec:combinatorial}, we consider the notion of substring complexity from the combinatorial point of view; and in Section~\ref{sec:finale}, we conclude this paper with a final remark on approximating $\delta$.

\section{Preliminaries}\label{sec:prel}
An {\em alphabet} $\Sigma$ is a finite nonempty set of elements called {\em letters}. We fix throughout a {\em string} $T=T[1]\cdots T[n]$ of {\em length} $|T|=n$ over an ordered alphabet $\Sigma$. 
By $\varepsilon$ we denote the \emph{empty string} of length $0$. 
For two indices $1 \leq i \leq j \leq n$, the \emph{$(i,j)$-fragment} of $T$ is an \emph{occurrence} of the underlying \emph{substring} $T[i\dd j]=T[i]\cdots T[j]$.
A {\em prefix} of $T$ is a fragment of $T$ of the form $T[1\dd j]$ and a {\em suffix} of $T$ is a fragment of $T$ of the form $T[i\dd n]$. A prefix (resp.~suffix) of $T$ is \emph{proper} if it is not equal to $T$. We let $T^r=T[n]T[n-1]\cdots T[1]$ denote the \emph{reversal} of $T$.

A positive integer $p$ is \emph{a period} of a string $T$ if 
$T[i]=T[j]$ whenever $i=j\pmod p$;
we call \emph{the period} of $T$, denoted by $\per(T)$, the smallest such $p$. 
A string $T$ is said to be \emph{strongly periodic} if $\per(T) \leq |T|/4$ and \emph{periodic} if $\per(T)\le |T|/2$.
We call the lexicographically smallest cyclic shift of $T[1\dd \per(T)]$ the \emph{(Lyndon) root} of $T$. Notice that if $T$ is periodic, then the root of $T$ is always a fragment of $T$ (that is, it has an occurrence in $T$).

For every string $t$ and every natural number $\ell$, we
define the $\ell$th \emph{power} of $t$, denoted by $t^{\ell}$, 
by $t^0=\varepsilon$ and $t^k = t^{k-1}t$,
for integer $k=[1,\ell]$.
A \emph{run} with (Lyndon) root $t$ in a string $T$ is a periodic fragment $T[i\dd j]=t[q\dd |t|]t^{\beta}t[1\dd \gamma]$, with $q,\gamma \in [1,|t|]$ and $\beta$ a positive integer, such that both $T[i-1\dd j]$ and $T[i\dd j+1]$, if defined, have their smallest period larger than $|t|$; we say that $q\in[1,|t|]$ is the \emph{offset} of the run $t[q\dd |t|]t^{\beta}t[1\dd \gamma]$ and that two runs with the same root are \emph{synchronized} if they have the same offset. We represent a run $t[q\dd |t|]t^{\beta}t[1\dd \gamma]$ by its starting and ending positions $(i,j)$ in $T$, its root $t$, and its offset $q$. 

The \textit{element distinctness} problem asks to determine if all the elements of an array $A$ of size $n$ are pairwise distinct. Yao showed that, in the comparison-based branching program model, the time required to solve the element distinctness problem using $\cO(b)$ space is in $\Omega(n^{2-o(1)}/b)$~\cite{DBLP:journals/siamcomp/Yao94}. We show the following lower bound for computing $\delta$ in the same model.

\begin{theorem}
The time required to compute $\delta$ for a string $T$ of length $n$ using $\cO(b)$ space in the comparison model is in $\Omega(n^{2-o(1)}/b)$.
\end{theorem}

\begin{proof}
We reduce the element distinctness problem to computing $\delta$ in $\cO(n)$ time as follows. Let $A$ be the input array for the element distinctness problem. 
Further let $\#_1,\#_2,\ldots,\#_n$ be pairwise distinct elements not occurring in $A$.
We set $T=A\cdot\#_1\#_2\ldots\#_n$, with $|T|=2n$, $\#_i\neq A[j]$, for all $i,j\in [1,n]$. Observe that $S_T(k)/k< n$, for all $k\geq 2$, and thus $\delta=S_T(1)=n+|\{A\}|$. Then $A$ has a repeating element if and only if $\delta<2n$. 
\end{proof}

\section{$\cO(n^3/b)$ Time Using $\cO(b)$ Space in the Comparison Model}\label{sec:simple}

We start with a warm-up lemma to guide the reader smoothly to the $\cO(n^3/b)$-time algorithm.

\begin{lemma}\label{lem:cubic}
    Given a string $T$ of length $n$, we can compute $\delta= \sup\{S_T(k)/k, k\geq 1\}$
    in $\cO(n^3)$ time using $\cO(1)$ space in the comparison model.
\end{lemma}

\begin{proof}

Let us consider each $S_T(k)$ separately, for all $k\in[1,n]$.

Set $S_T(k)=0$. For all $i\in[1,n]$, we increase $S_T(k)$ if
$T[i\dd i+k-1]$ is the first occurrence in $T[1\dd i+k-1]$.
To perform this we check whether $T[j\dd j+k-1]=T[i\dd i+k-1]$, for all $j\in [1,i-1]$. We employ any linear-time constant-space pattern matching algorithm~\cite{DBLP:journals/jcss/GalilS83,DBLP:journals/jacm/CrochemoreP91,DBLP:journals/tcs/BreslauerGM13} to do this check in $\cO(n)$ time using $\cO(1)$ space for a single $i$. The statement follows.
\end{proof}

We next generalize Lemma~\ref{lem:cubic} by employing
the following straightforward observation.

\begin{observation}\label{obs:occur}
Let $S$ be a substring of $T$. If $S$ occurs at least twice in $T$, then every substring of $S$ occurs at least twice in $T$; if $S$ occurs only once in $T$, then any substring of $T$ containing $S$ as a substring occurs only once in $T$.
\end{observation}

\subparagraph*{Main Idea.}~Recall that we have $\cO(b)$ budget for space. At any phase of the algorithm, we maintain $S_T(k)$ for $b$ values of $k$, and iterate on consecutive non-overlapping substrings of $T$ of length $b$, which we call {\em blocks}. This gives $n/b$ phases and $n/b$ iterations per phase, respectively. For each iteration, we define a substring $M$ of $T$, which we call {\em anchor}. We search for occurrences of this anchor in $T$ and extend each of the (at most) $n$ occurrences of $M$ in $\cO(b)$ time per occurrence. This gives $\cO(n^3/b)$ time and $\cO(b)$ space.

\begin{proposition}\label{pro:simple}
    Given a string $T$ of length $n$, we can compute $\delta= \sup\{S_T(k)/k, k\geq 1\}$
    in $\cO(n^3/b)$ time using $\cO(b)$ space, for any $b\in[1,n]$, in the comparison model.
\end{proposition}

\begin{proof}
Our algorithm consists of $n/b-2$ phases.
In phase $\alpha$, for all $\alpha\in[2,3,\ldots,n/b-1]$,\footnote{We process the substrings of length $k\in[1,2b]$ separately: for each block $B_j$,
we compute an array $\textsf{LS}_j$ of size $b$ such that $\textsf{LS}_j[q]$ is the length 
the {\em longest substring of length up to $2b$} starting at position $q$ in $B_j$ that occurs in $T$ before position $(j-1)b+q$.
This is done as described in the proof of Lemma~\ref{lem:smallk} and requires $\cO(\frac{n^2\log b}{b})$ total time. At the end of this procedure, we just maintain $\max\{S_T(k)\mid k\in [1,2b]\}$.} we compute altogether the $b$ values of $S_T(k)$, for all $k\in [\alpha b +1,(\alpha+1) b]$. Let $\mathcal{S}=\mathcal{S}[1\dd b]$ be an array of size $b$ where we store the values of $S_T(k)$ corresponding to phase $\alpha$: $\mathcal{S}[h]=S_T(\alpha b + h)$, for $h\in[1,b]$.  At the end of phase $\alpha$ we maintain the maximum of $\mathcal{S}[h]/(\alpha b+h)$.
Clearly, at the end of the whole procedure, we can output $\delta= \sup\{S_T(k)/k~|~k\geq 1\}$.

We start by decomposing $T$ into $n/b$ blocks $B_1,B_2,\ldots,B_{n/b}$, each of length $b$.
We next describe our algorithm for a fixed phase $\alpha>1$. First we set $\mathcal{S}[h]=0$, for all $h\in[1,b]$. Let $i$ be a position on $T$. For each $k$ in the range of $\alpha$, we want to know if $T[i\dd i+k-1]$ has its first occurrence in $T$ at position $i$ or if it occurs also at some position to the left of $i$.
We process together all positions $i$ in the same block $B_j=T[(j-1)b+1\dd jb]$, for every $j\in[1,n/b]$. Let $L=B_j$ be the block we are currently processing (inspect also Figure~\ref{fig:algo_fig}).
To compute $S_T(k)$ we consider, for all $k\in [\alpha b+1,(\alpha+1) b]$, the length-$k$ fragments with starting position $i$ in $L$. All such fragments share the same \emph{anchor} $M=T[jb+1\dd  (j+\alpha-1)b]$. The fragment of length $k$ ends at position $i+k-1$, which belongs to one of the {\em two blocks} succeeding $M$ for all $k\in [\alpha b+1,(\alpha+1) b]$; we denote the concatenation of these two succeeding blocks as fragment $R$. In particular, we have $|M|=(\alpha-1)b$ and $R=B_{j+\alpha}B_{j+\alpha+1}$.

\begin{figure}[t]
    \centering
    \includegraphics[width=.8\linewidth]{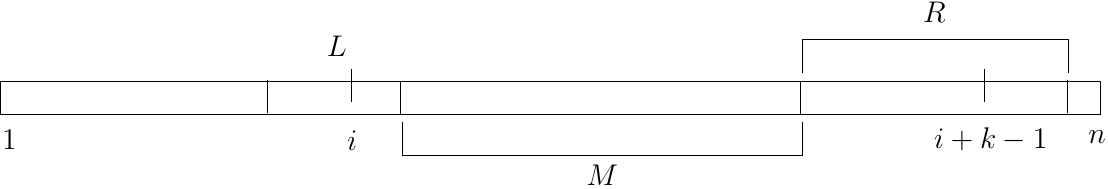}
    \caption{The main setting of the algorithm underlying Proposition~\ref{pro:simple}.}
    \label{fig:algo_fig}
\end{figure}

We will use the occurrences of $M$ in $T$ that start before its {\em starting position} $jb+1$ as anchors for finding possible occurrences of the length-$k$ fragments starting within $L$.
We search for such occurrences of $M$ with any linear-time constant-space pattern matching algorithm~\cite{DBLP:journals/jcss/GalilS83,DBLP:journals/jacm/CrochemoreP91,DBLP:journals/tcs/BreslauerGM13}. 
For each such occurrence of $M$, we then need to check the $b$ letters preceding it and the $2b-1$ letters following it in order to determine whether it generates a previous occurrence of some $(i,i+k-1)$-fragment, where $i$ is a position within the block $L$.
In particular, we check the $b$ letters preceding it because $L$ is the block of $b$ positions preceding $M$; we check the $2b-1$ letters following it because $k\leq(\alpha+1)b$.

While processing $L= B_j = T[(j-1)b+1\dd jb]$, we also maintain an array $\textsf{END}_L[1\dd b]$ of size $b$. After we have finished processing $L$, $\textsf{END}_L[q]$ will store the length $r_q$ of the longest prefix of $R$ such that $T[(j-1)b+q\dd (j-1)b+q+|M|+r_q-1]$ occurs in $T$ before position $(j-1)b+q$ (inspect Figure~\ref{fig:lcp}).
We compute $\textsf{END}_L$ as follows. We search for all the occurrences of $M=T[jb+1\dd  (j+\alpha-1)b]$ in $T[1\dd  (j+\alpha-1)b-1]$, from left to right. Let $M=T[i'\dd i'+|M|-1]$ be one such occurrence. Let $\ell$ be the length of the longest common suffix of $L$ and $T[1\dd i'-1]$; let $r$ be the length of the longest common prefix of $R$ and  $T[i'+|M| \dd n]$.
For each $q\ge b-\ell+1$, we update $\textsf{END}_L[q]$ with the maximum between its previous value $ \textsf{END}_L[q]$ and $r$ (note that we do not update any values if $\ell=0$).
After we have processed all the occurrences of $M$, for each $q$ we increase by $1$ all $S_T(k)$ such that $k> b-q+1+|M|+ \textsf{END}_L[q]=b-q+1+(\alpha-1)b+ \textsf{END}_L[q]=\alpha b-q+ \textsf{END}_L[q]+1$. 
This is an application of Observation~\ref{obs:occur}: all these occurrences correspond to a substring that is longer than a substring that occurs {\em for the first time} in $T$ at position $(j-1)b+q$. 

\begin{figure}[t]
    \centering
    \includegraphics[width=.8\linewidth]{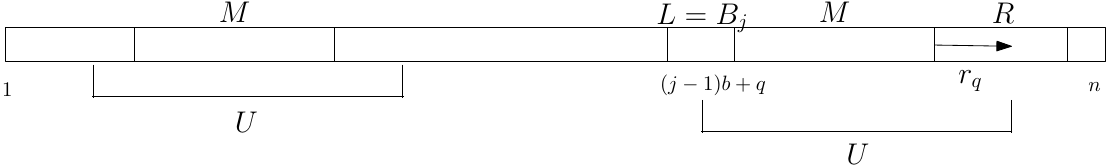}
    \caption{Largest $r_q$ such that $U=T[(j-1)b+q\dd (j-1)b+q+|M|+r_q-1]$.}
    \label{fig:lcp}
\end{figure}

The whole algorithm takes time $\cO(\frac{n^3}{b})$:
there are $\frac{n}{b}$ phases; in each phase, we consider $\frac{n}{b}$ blocks and for each block we spend $\cO(n)$ time for pattern matching anchor $M$; for each occurrence of the anchor, we spend $\cO(b)$ time for finding and updating the possible extensions, thus $\cO(nb)$ time overall. We finally need $\cO(b^2)$ time for updating the values of $S_T(k)$ for all $k$'s in the range and all positions $i$ in $L$. Overall this is $\cO(\frac{n}{b}\cdot \frac{n}{b} \cdot (nb+b^2))=\cO(\frac{n^3}{b})$ time.
\end{proof}

\section{$\cO(\frac{n^3\log b}{b^2})$ Time Using $\cO(b)$ Space in the Comparison Model}\label{sec:improve}

Recall that in Proposition~\ref{pro:simple}, we spend $\cO(nb+b^2)$ time to process the at most $n$ occurrences of a single anchor $M$ in $T$. We show here that all these occurrences can be processed in $\cO(n\log b)$ time. This is made possible by processing together batches of occurrences of $M$ that are {\em close enough} in $T$. This is done by means of answering longest common extension queries on suffix trees constructed for certain length-$\cO(b)$ fragments of $T$. 

The first trick is based on the following remark: The pattern matching algorithm for reporting the occurrences of $M$ (e.g.,~\cite{DBLP:journals/tcs/BreslauerGM13}) reports the occurrences of $M$ in real-time from left to right. Every such occurrence $m$ of $M$ is preceded by a block $L'$ of length $b$ on the left of $m$ starting at position $m-b$ and ending at position $m-1$, and it is succeeded by a fragment $R'$ of length $2b$ starting at position $m+|M|$ and ending at position $m+|M|+2b-1$. We thus need to find the longest common prefix of $R$ and $R'$ and the longest common suffix of $L$ and $L'$. 
Let us describe the process for the longest common prefix of $R$ and $R'$. (The procedure for the longest common suffix of $L$ and $L'$ is analogous and is executed simultaneously.) 

We use the so-called {\em standard trick} to construct a sequence of $n/(4b)$ suffix trees for fragments of $T$ of length $4b$ overlapping by $2b$ positions. We first concatenate each such fragment of length $4b$ with $R$. Constructing one such suffix tree takes $\cO(b\log b)$ time using $\cO(b)$ space~\cite{DBLP:conf/focs/Weiner73}. Recall that an occurrence $m$ of $M$ implies an occurrence of $R'$ at position $m+|M|$ and thus this position is part of some fragment of length $4b$. We preprocess this suffix tree in $\cO(b)$ time and space to answer longest common prefix queries in $\cO(1)$ time~\cite{DBLP:conf/latin/BenderF00}. The whole preprocessing thus takes $n/(4b)\cO(b\log b)=\cO(n\log b)$ time. Thus, for any occurrence $m$ of $M$ we can find the longest right extension (and the longest left extension with a similar procedure) in $\cO(1)$ time; recall that each extension cannot be of length greater than $2b$ so we do not miss any of them. To memorize the extensions we use an array $\textsf{END}_L$  of size $b$. For each occurrence of $M$, if we have a left extension of length $\ell>0$ and a right extension of length $r$, we set $\textsf{END}_L[b-\ell+1]=\max\{\textsf{END}_L[b-\ell+1],r\}$ in $\cO(1)$ time. At the end of this process we sweep through $\textsf{END}_L$ and set $\textsf{END}_L[q]=\max\{\textsf{END}_L[q-1],\textsf{END}_L[q]\}$, for all $i\in[2,b]$ by Observation~\ref{obs:occur}: if we can extend a position $q$ in $L$ $r$ positions to the right of $M$, then we must be able to extend position $q+1$ in $L$ at least $r$ positions to the right of $M$.

The second trick updates all values of $S_T(k)$ using array $\textsf{END}_L$ in $\cO(b)$ time instead of $\cO(b^2)$ time. We use an array $I$ of size $b$ with all its entries initialized to $0$; $I[h]$ will store the number of positions $q$ in $L$ such that the shortest unique substring starting at $q$ is of length $\alpha b + h$. We fill in $I$ scanning $\textsf{END}_L$: the shortest unique substring starting at $q$ is by definition of length $\alpha b-q+ \textsf{END}_L[q]+2$, which equals $\alpha b+h$ when $h=\textsf{END}_L[q]-q+2$.
We thus increment $I[h]$ by one. We finally increase $S_T(\alpha b+h)$ by $\sum_{j=1}^h I[j]$ for all $h=1,\ldots,b$. Thus, updating all values of $S_T(k)$ is implemented in $\cO(b)$ time. 
We have arrived at Theorem~\ref{the:n3}.

\section{$\tilde{\cO}(\frac{n^2}{b})$ Time Using $\tilde{\cO}(b)$ Space for $b\geq  \sqrt{n}$ in the word RAM model}\label{sec:main}
The algorithm underlying Theorem~\ref{the:n3} is organized in $n/b$ phases. In phase $\alpha$ we process $b$ values of $S_T(k)$ making use of evenly-spaced fragments of $T$, each of length $(\alpha-1)b$, as \emph{anchors} for finding possible multiple occurrences of the length-$k$ fragments of $T$.
Considering $\cO(n/b)$ anchors in each phase and processing them one by one is the bottleneck of this algorithm. Our approach here is thus to avoid the burden of considering new anchors at every phase by carefully selecting a set of anchors that will remain \emph{unchanged} in each phase of the algorithm. Let $c>1$ be any integer constant. We will process the values of $S_T(k)$ for $k\leq cb$ (Section~\ref{sec:smallk}) and for $k>cb$ (Sections~\ref{sec:regions} to~\ref{sec:mainlargek}) in two different ways.

We work in the word RAM model and our goal is a deterministic algorithm. 
Recall that a suffix tree of any string of length $d$ can be constructed in $\cOtilde(d)$ time using $\cO(d)$ space~\cite{DBLP:conf/focs/Weiner73,DBLP:conf/focs/Farach97}.

\subsection{Computing $S_T(k)$ for Small $k$}\label{sec:smallk}
We process together all values $k\in[1,cb]$.
Like in Sections~\ref{sec:simple}-\ref{sec:improve}, for such values of $k$ we split $T$ into $n/b$ blocks of $b$ positions and work with each such block separately; we compute all values $S_T(k)$ and keep track of $\max_{k\le cb}S_T(k)/k$ before computing $S_T(k)$ for all $k>cb$.

Consider block $L= B_j = T[(j-1)b+1\dd jb]$.
We compute an array $\textsf{LS}_L$ of size $b$ such that $\textsf{LS}_L[q]$ is the length 
the {\em longest substring} starting at position $q$ in $L$ that occurs in $T$ before position $(j-1)b+q$, if this length does not exceed $cb$, otherwise we set it to $\infty$.
This is done by constructing multiple generalized suffix trees of windows of length $2cb$ and $L$.

\begin{restatable}{lemma}{lemsmallk}
\label{lem:smallk}
$\max\limits_{k\le cb}\frac{S_T(k)}{k}$ can be computed in $\tilde{\cO}(n^2/b)$ time and $\cO(b)$ space, for any $b\in[1,n]$.
\end{restatable}
\begin{proof}
We consider a block $L=T[(j-1)b+1\dd jb]$ of $b$ positions of $T$ at a time; for each position $i$ of $T$ within $L$, we must compute the length of the longest fragment $T[i\dd \ell]$ that occurs to the left of position $i$, if this length does not exceed $cb$.
We consider windows of length $2cb$ over the prefix $T[1\dd (j+c)b]$, overlapping by $cb$ positions. Clearly, if a fragment $T[i\dd \ell]$ occurs earlier in $T$, then it must be a substring of at least one such window.
For a fixed $L$ we initialize all the $b$ positions of an array $\textsf{LS}_L$ to $0$; we then consider one window $W$ of $2cb$ positions at a time, from left to right. At the end of the computation for a window $W$, $\textsf{LS}_L[q]$ will store the length of the longest fragment starting at position $(j-1)b+q$ which occurs earlier in $T$. We proceed as follows to achieve this computation.

For the current window $W$ of length $2cb$, we concatenate $W$ and $T[(j-1)b+1\dd (j+c)b]=L\cdot T[jb+1\dd (j+c)b]$ (that is, block $L$ and the following $cb$ positions)
constructing a new string $S$; we use a separator letter that does not occur in either of the two strings. We then construct the suffix tree of $S$;
and from there on the Longest Previous Factor (LPF) array of $S$ in $\cO(|S|)=\cO(b)$ time~\cite{DBLP:journals/ejc/CrochemoreIIKRW13}.
The LPF array is an array of length $|S|$; for each position $i$ of $S$, it gives the length of the longest substring of $S$ that occurs both at $i$ and to the left of $i$ in $S$. 
Finally, we use this information to update the values of $\textsf{LS}_L$: $\textsf{LS}_L[q]$ maintains the maximum between its previous value and the new value computed for the current $W$. We proceed to the next window. Once we have processed all the windows, we use $\textsf{LS}_L$ to update the corresponding values of $S_T$ in $\cO(b)$ time the same way as we used $\textsf{END}_L$ in Section~\ref{sec:improve}.

The time and space complexity is as follows.
There are $n/b$ blocks in $T$, each of length $b$. For each such block, we consider $\cO(n/b)$ windows of $2cb$ positions each, and for each window, we construct the suffix tree and the LPF array of the two underlying fragments of length $\cO(b)$ in $\cOtilde(b)$ time using $\cO(b)$ words of space. The whole procedure, for all $n/b$ blocks and all $\cO(n/b)$ windows, thus requires $\cOtilde(\frac{n}{b}\frac{n}{b} b)=\cOtilde(n^2/b)$ time using $\cO(b)$ words of space.
\end{proof}

\subsection{$b$-Runs and $b$-Gaps}\label{sec:regions}
When $k>cb$, we process $b$ values of $S_T(k)$ at each phase, just like we did in Section~\ref{sec:improve}. 
Different from Section~\ref{sec:improve}, though, we aim at selecting a \emph{global} set of anchors, carefully chosen among the length-$b$ substrings of $T$.
At each phase, we will distinguish three types of substrings, depending on the period of their length-$b$ substrings. 
A $b$-\emph{run} is a maximal fragment of length at least $b$ such that each of its length-$b$ substrings is strongly periodic; 
a standard reasoning based on the periodicity lemma~\cite{PeriodicityLemma} shows that the period of each $b$-run is at most $b/4$, and so a $b$-run is indeed a run.
A $b$-\emph{gap} is a maximal fragment such that none of its length-$b$ substrings is strongly periodic.
Any fragment of $T$ of length at least $b$ and period at most $b/4$ is fully contained in a unique $b$-run; and every fragment of $T$ of length at least $b$ and such that none of its length-$b$ substrings is strongly periodic is fully contained in a unique $b$-gap.
At each phase, the substrings to be processed are thus of three types: (i) either they are fully contained in a $b$-gap, or (ii) they are fully contained in a $b$-run, or (iii) neither of the two.
We will process the substrings differently depending on their type.
A standard reasoning using the periodicity lemma~\cite{PeriodicityLemma} shows that two $b$-runs cannot overlap
by more than $b/2$ letters, so there are only $\cO(n/b)$ of them.
Lemma~\ref{lem:regions} states that we can identify and store the $b$-runs of $T$ in such space complexity.
For proving it we rely on the space-efficient construction of sparse suffix trees. The term ``sparse'' refers to constructing the compacted trie of an arbitrary subset of the set of the suffixes of the input string.

\begin{theorem}[\cite{DBLP:conf/soda/BirenzwigeGP20}]\label{the:sparseST}
Given a set $B\subseteq[n]$ of size $\Omega(\log n)\leq |B|\leq n$, there exists a deterministic algorithm which constructs the (sparse) suffix tree of $B$ in $\cO(n\log \frac{n}{|B|})$ time using $\cO(|B|)$ words of space.
\end{theorem}

\begin{restatable}{lemma}{lemregions}\label{lem:regions}
A representation of the $b$-runs of $T$ can be computed in $\cOtilde(n)$ time using $\cO(n/b)$ space, which is $\cO(b)$ space when $b\geq \sqrt{n}$.
\end{restatable}
\begin{proof}
We process windows of $b$ positions of $T$ at a time, with any two consecutive windows overlapping by $b/2$ positions. At each step, we compute the longest suffix, which has period at most $b/4$, of the window in $\cO(b)$ time~\cite{crochemore2007algorithms}. If such a suffix has nonzero length, we keep track of its starting position in $T$ and extend it na\"{i}vely to the right as much as possible. If this extension results in a run of length at least $b$, we store its starting and ending position in a list ordered by starting position and resume the process using the window starting $b-1$ positions before the end of the run. Otherwise, if the extension results in a run shorter than $b$, we ignore it. Whenever we identify a $b$-run, we compute its root $t$ in $\cO(b)$ time~\cite{duval1983factorizing}, and store in a list the starting and ending position $(s_r,e_r)$ of its root and the starting and ending position $(s,e)$ of the $b$-run (as mentioned above). After computing all $b$-runs in $T$, we construct the sparse suffix tree over the set of all $s_r$ positions in the list. 
Each internal node of the sparse suffix tree, corresponding to a root of a $b$-run of $T$, is associated with the list of the starting and ending positions $(s,e)$ of the $b$-runs corresponding to this root.

This procedure identifies all the $b$-runs of $T$. Indeed, consider a window $T[i\dd i+b-1]$. If a $b$-run $Y$ with period $p\leq b/4$ begins between position $i$ and position $i+b-1-p$, a prefix of it of length greater than $p$ is a suffix of the window with period $p$. If it is the longest such suffix, it will be extended to the right allowing the identification of the whole $Y$. Otherwise, suppose there is a longer suffix of $T[i\dd i+b-1]$ with period $b/4\ge p'> p$ (it cannot be $p'<p$, because otherwise, $p'$ would have been the period of the whole suffix) that includes the whole prefix of $Y$ in $T[i\dd i+b-1]$. 
In this case, we only extend the longer suffix and do not find $Y$ at this stage. However, the longer suffix with period $p'\leq b/4$ is part of a run that overlaps with $Y$, and therefore such overlap must be shorter than $b/2$ because of the periodicity lemma~\cite{PeriodicityLemma}. This means: (a) this situation can only happen when the prefix of $Y$ in $T[i\dd i+b-1]$ is shorter than $b/2$, thus a longer prefix of $Y$ will be a suffix of the next window $T[i+b/2\dd i+3b/2-1]$; and (b) the period $p'$ must break before the end of $T[i+b/2\dd i+3b/2-1]$, thus the prefix of $Y$ in $T[i+b/2\dd i+3b/2-1]$ must be the longest suffix with period at most $b/4$ and will therefore be extended, allowing to identify the whole $Y$.
Finally, if $Y$ begins between position $i+b-p$ and position $i+b-1$ of $T[i\dd i+b-1]$, its prefix included in the window does not have a period $p$, and will therefore not be extended. However, the next window is $T[i+b/2\dd i+3b/2-1]$: since the length of any $b$-run is at least $b$, a prefix of the $b$-run of length greater than $b/2$ is now a suffix of the window, and since $p\leq b/4$, it will be extended to the right allowing the identification of the whole $b$-run. 

The time and space complexity is as follows. We consider $\cO(n/b)$ windows of length $b$. At each step, we spend $\cO(b)$ time to compute the longest suffix of the current window with period at most $b/4$. Whenever we identify a suffix of a run $Y$ with period at most $b/4$, we extend it na\"{i}vely to the right in $\cO(|Y|)$ time, and the next window we consider only covers the last $b-1$ positions of $Y$. Since consecutive $b$-runs can only overlap by less than $b/2$ positions because of the periodicity lemma~\cite{PeriodicityLemma}, they are at most $\cO(n/b)$ and their total length is $\cO(n)$, so it takes $\cO(n)$ time to perform all extensions. 
For each $b$-run, we spend $\cO(b)$ time to compute its root.
For the sparse suffix tree, we employ Theorem~\ref{the:sparseST}.
Hence the overall time complexity is $\cOtilde(n)$. As for the space, we process blocks of $\cO(b)$ positions in $\cO(b)$ space. We also store a pair of positions for each $b$-run, therefore the space required to store them is $\cO(n/b)$, which is $\cO(b)$ when $b\ge\sqrt{n}$. 
\end{proof}
The output of Lemma~\ref{lem:regions} is a list representing all the $b$-runs of $T$ in the natural left-to-right order. The $b$-gaps can be deduced from this list
as follows: if $T[i\dd j]$ and $T[i'\dd j']$ are two consecutive $b$-runs in the list, then $T[j-b+2\dd i'+b-2]$ is a $b$-gap (if $T[i\dd j]$
is the first run, then so is $T[1\dd i+b-2]$, and similarly for the last run).

A subset of the length-$b$ substrings of $T$ is a \emph{valid set of anchors} if two properties hold: (i) at least one anchor occurs in each fragment of $T$ of length $cb$; and (ii) the total number of occurrences of all anchors in $T$ is in $\cO(n/b\cdot\log n)$. Lemma~\ref{lem:anchorexist} shown next will be useful to prove that there always exists a set of valid anchors included in the $b$-gaps of $T$.  


\begin{restatable}{lemma}{lemanchorexist}\label{lem:anchorexist}
Let $Z$ be a string with all length-$d$ substrings not strongly periodic, and $c>1$ be any integer constant. Then we can compute 
in $\cOtilde(|Z|^2/d)$ time and $\cOtilde(|Z|/d+d)$ space
a subset $A$ of the length-$d$ substrings of $Z$ such that: (i) at least one $h\in A$ occurs in each fragment of $Z$ of length $cd$; and (ii) the total number of occurrences of all $h\in A$ in $Z$ is $\cO(|Z|/d\cdot\log |Z|)$.
\end{restatable}
\begin{proof}
Let us start with a high-level idea of the proof. 
We first reduce the problem to the following: we have $\cO(|Z|/d)$ strings $Z_{i}$, each of length
$5d/4$ and with all length-$d$ substrings not strongly periodic, and a set of $\cO(|Z|)$ possible anchors consisting of all length-$d$
substrings of the $Z_{i}$s.
We want to choose a subset $A$ of the anchors such that (i) at least one $h\in A$ occurs in each $Z_{i}$, (ii) the total number of
occurrences of all $h\in A$ in the $Z_{i}$s is $\cO(|Z|/d\cdot \log |Z|)$.
This is a special case of the Node Selection problem, considered in~\cite{bernardini2021elasticdegenerate}
 as a strengthening of the well-known Hitting Set problem.\footnote{Let us remark that this problem has already been considered in the conference
 version~\cite{DBLP:conf/icalp/0001GPPR19}, with a slightly different definition but essentially the same proof. However, our goal
 is a deterministic algorithm and to this end we need~\cite{bernardini2021elasticdegenerate}, the extended version of~\cite{DBLP:conf/icalp/0001GPPR19}.}
 Indeed, we can take $U$ to be the set of strings $Z_{i}$, $V$ to be the set of possible
anchors, and add an edge $(u,v)$ in $G(U,V,E)$ when the possible anchor corresponding to $v$ occurs in the string $Z_{i}$ corresponding to $u$.
Because every possible anchor is not strongly periodic and every $Z_{i}$ is of the same length $5d/4$, 
the degree of every node $u\in U$ is $5d/4$. Then, by Lemma 5.4 of~\cite{bernardini2021elasticdegenerate} (the weights are irrelevant)
we can choose a set $V'\subseteq V$ such that (i) $N[u]\cap V' \neq \emptyset$ for every $u\in U$,
(ii) $\sum_{u\in U} |N[u]\cap V'| = \cO(|U|\log |U|) = \cO(|Z|/d \cdot \log |Z|)$, so $V'$ corresponds to a set
of anchors $A'$ with the sought properties.
Furthermore, $V'$ can be found in linear time and space in the size of $G$, which is $\cO(|Z|)$. This is however not enough
for our purposes, as we cannot store the whole $G$. Analysing the algorithm used inside
the proof of Lemma 5.4 of~\cite{bernardini2021elasticdegenerate} we see that it considers the nodes $v\in V$ one-by-one
while maintaining some information of size $\cO(|U|)=\cO(|Z|/d)$ and a precomputed table of a size that can be bounded by
the maximum degree of any $u\in U$, which is $\cO(d)$. Furthermore, the algorithm accesses $G$ only by iterating a constant number
of times over the neighbours of the current node $v\in V$.
In what follows, we show how to implement this efficiently in our model to achieve the claimed bounds.

The first step is to obtain the strings $Z_{i}$. Consider windows of $5d/4$ positions of $Z$ overlapping by $d$ positions,
and let $Z_{i}$ be the $i$th such window from left to right.
Note that, since the starting positions of any two consecutive windows are $d/4$ positions apart, there are $\cO(|Z|/d)$ such windows in $Z$.
We claim that selecting a set of anchors, each of length $d$, such that at least one anchor occurs in each $Z_i$, implies property (i).
This is because $Z_i$ and $Z_{i+1}$ span together $5d/4+d/4=3d/2$ positions of $Z$, and thus any fragment of $Z$ of length $cd\geq 2d$
fully contains at least one window $Z_i$.
We thus aim at selecting a set $A$ of length-$d$ substrings of $Z$ (anchors) such that at least one $h\in A$ occurs in $Z_i$ for all $i$
and such that the total number of occurrences of all anchors in all windows $Z_i$ is $\cO(|Z|/d\cdot\log |Z|)$. Since any position of $Z$
belongs to a constant number of windows, the second requirement is enough to guarantee that property (ii) holds.

The next step is to simulate iterating over the nodes $v\in V$. We iterate over windows of $5d/4$ positions of $Z$ overlapping
by $d$ positions from left to right. For the current window $W$ we need to check which of the length-$d$ fragments of $W$ are
their leftmost occurrences. To compute this information we construct the suffix tree $\textsf{ST}(W)$ of each window $W$ in time
$\cOtilde(d)$ and space $\cO(d)$.
 We search each length-$d$ substring of the prefix of $Z$ up to the end of $W$ in $\textsf{ST}(W)$ in time $\cOtilde(d)$, and we mark the nodes corresponding to the ones we find.
This can be done by scanning with a window of length $d$ and maintaining the longest prefix of the current window with
a corresponding node (implicit or explicit) in the suffix tree. After moving the window by one to the right, we follow the suffix link (if we are
at an implicit node we use the suffix link of its nearest explicit ancestor) and then possibly descend down; this takes amortised constant time.
We then consider each length-$d$ substring $S$ of $W$ whose corresponding node in $\textsf{ST}(W)$ is not marked in the natural left-to-right order, while maintaining the corresponding node
of the suffix tree (we are guaranteed that such a node exists) as explained above.
This gives us the information about the length-$d$ fragments of $W$ that should be considered as nodes $v\in V$,
so that we can iterate over them efficiently.

Having the length-$d$ fragment $S$ corresponding to the current node $v\in V$, we need to simulate generating the neighbours
of $v\in V$ in $G$, which translates into generating the occurrences of $S$ in all fragments $Z_{i}$.
It is enough to implement this step in $\cO(|Z|/d)$ time (even when there are very few occurrences) as this will sum up to $\cO(|Z|^{2}/d)$.
We observe that every $S$ in the current window $W$ includes the same fragment of length $3d/4$, namely $m=W[d/4\dd d-1]$,
so that $W$ can be written as $x\cdot m\cdot y$. This means that any occurrence of any such length-$d$ substring $S$ in $Z_i$
must be in correspondence with an occurrence of $m$ in $Z_i$.
If $m$ is not strongly periodic, we can compute and store all of its occurrences in every $Z_{i}$ in total $\cO(|Z|)$ time and $\cO(|Z|/d)$
space with a linear-time constant-space pattern matching algorithm. Additionally, for each such occurrence $q$ of $m$ in $Z_{i}$ (there are $\cO(1)$ of
them because $m$ is not strongly periodic and $|Z_{i}|=5d/4$)
we store the following information. We write $Z_{i}$ as $x'_q \cdot m \cdot y'_q$ and compute in $\cO(d)$ time the longest common suffix of $x$ and $x'_q$
and the longest common prefix of $y$ and $y'_q$.
Then, to check if the current $S$ occurs in $Z_{i}$, we write $S$ as $x'' \cdot m \cdot y''$ and iterate over the occurrences $q$ of $m$
in $Z_{i}$.  
Then we check if the stored longest common suffix for occurrence $q$ is of length at least $|x''|$
and the stored longest common prefix is of length at least $|y''|$. Thus, for each $S$ we can check if it occurs in $Z_{i}$
in $\cO(1)$ time. Together with the preprocessing done for each window, this sums up to $\cO(|Z|^{2}/d)$ as claimed.

Finally, we need to explain what to do when $m=W[d/4\dd d-1]$ is strongly periodic. In such case, we find the rightmost position
$i_1\leq d/4$ such that the period of $W[i_1\dd i_1+3d/4-1]$ is larger than $d/4$. Similarly, we find the leftmost position $i_2\geq  d/4$
such that the period of $W[i_2\dd i_2+3d/4-1]$ is larger than $d/4$. By appending and prepending special letters to $W$, we
can assume that both $i_1$ and $i_2$ are well-defined. By a standard argument based on the periodicity lemma,
see e.g.~\cite[Lemma 4]{bernardini2021elasticdegenerate}
the period of the whole $W[i_1+1 \dd i_2-1]$ is at most $d/4$. Therefore, any length-$d$ fragment must intersect
position $i_1$ or $i_2$ as all such fragments are not strongly periodic. Consequently, we can consider
$m_{1}=W[i_1\dd i_1+3d/4-1]$ and $m_{2}=W[i_2-3d/4 +1 \dd i_2]$.
Both $m_{1}$ and $m_{2}$ are not strongly periodic by the choice of $i_1$ and $i_2$, and because $|W|=5d/4$ any length-$d$
substring of $W$ contains either $m_{1}$
or $m_{2}$ fully inside. Thus, we can apply the reasoning from the above paragraph twice.
\end{proof}

\subsection{Processing the $b$-Gaps}\label{sub:baperiodic}
For ease of presentation, in this section, we will assume that all length-$b$ substrings of $T$ are not strongly periodic, but
no major changes are required to apply the same reasoning on the set of all $b$-gaps.
Assume we have already computed a set $A$ of valid anchors over $T$.
For each $h\in A$, we compute a list of its occurrences in $T$. The overall size of these lists is $\cO(n/b\cdot\log n)$
because of property (ii), and the occurrences of each $h\in A$ can be generated in $\cO(n)$ time and $\cO(1)$ space
(plus the space to store the list) with any linear-time constant-space pattern matching algorithm, so $\cOtilde(n^{2}/b)$ time overall.
We divide the computation of $S_T(k)$ in $n/b$ phases.
Consider phase $\alpha$, in which we consider substrings of length $k\in [\alpha b+1,(\alpha+1)b]$. Because of property (i), at least one anchor occurs in the first $cb$ positions of each such substring.
We conceptually associate such a substring with the leftmost anchor $h\in A$ occurring therein, and we say that a fragment of $T$ is \emph{anchored} at an occurrence $i$ of some anchor $h$ if the leftmost occurrence of any anchor in the fragment is $i$. We then process the substrings according to the anchor with which they are associated.

All substrings associated with an anchor $h\in A$ have a (possibly empty) prefix of length $\cO(b)$ where no anchors occur, followed by $h$ and then by a suffix where any anchor can occur. This implies that any occurrence of such substrings can only start in a range of $\cO(b)$ positions preceding some occurrence of $h$ in $T$. In particular, if $h$ occurs at position $i$ in $T$ and the closest anchor to its left is at position $i'<i$, the \emph{starting range} of substrings of $T$ associated with $h$ is $[i'+1, i]$, or $[1,i]$ if $i$ is the first occurrence of any anchors in $T$. All starting ranges for all anchors can be computed in $\cO(n)$ time by scanning the list of occurrences of the anchors.
To update the values of $S_T(k)$ with the substrings associated with $h$ we need to know, for each occurrence $i$ of $h$ in $T$ and each of its previous occurrences $i'<i$, the longest left extension within the starting ranges of $i$ and $i'$, 
and the longest right extension of the fragments of $T$ following the occurrences of $h$ at $i$ and $i'$. 
We cannot afford to store all these pairs of values explicitly as this would require $\cOtilde(n^2/b^2)$ space. We thus construct a separate data structure, denoted by $D(h)$, for each anchor $h\in A$. This data structure encode the same information in a compact form. We next describe the data structure and its construction.

\begin{figure}[t]
    \centering
    \includegraphics[width=.8\linewidth]{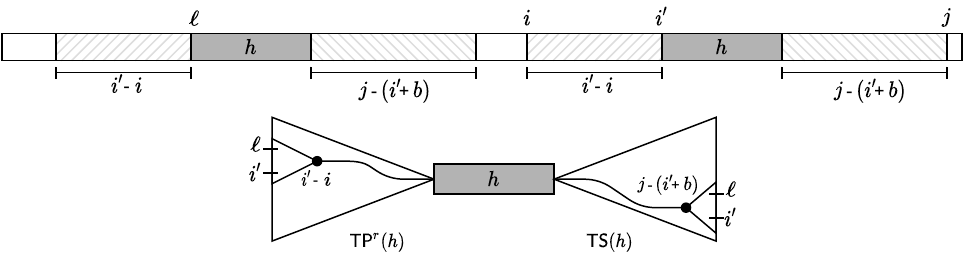}
    \caption{A previous occurrence of $T[i\dd j]$ anchored at $i'$ can be detected using $D(h)$.}
    \label{fig:D(h)}
\end{figure}

$D(h)$ consists of two \emph{compacted} tries $\textsf{TP}^{r}(h)$ and $\textsf{TS}(h)$. 
For every occurrence $i$ of $h$ in $T$, $\textsf{TS}(h)$ contains a leaf corresponding to $T[i+b\dd n]$, and $\textsf{TP}^{r}(h)$ a leaf corresponding to $(T[1\dd i-1])^{r}$, both labelled with position $i$.
We only store the list of children and the length of the path label of each node, which we call its \emph{depth}. Because of property (ii), the overall size of these data structures for all anchors is thus in $\cO(n/b\cdot \log n)$.
For any two occurrences $i,i'$ of $h$, the depth of the lowest common ancestor of leaves $i$ and $i'$ in $\textsf{TP}^{r}(h)$ gives the length of their longest left extension, and the depth of their lowest common ancestor in $\textsf{TS}(h)$ gives the length of their longest right extension: see Figure~\ref{fig:D(h)} for an example.
$D(h)$ can be efficiently constructed for all $h\in A$, as shown by Lemma~\ref{lem:D(h)}. 

\begin{restatable}{lemma}{lemDh}\label{lem:D(h)}
Data structures $D(h)$, for all $h\in A$, can be constructed in $\tilde{\cO}(n^{2}/b)$ total time using $\tilde{\cO}(n/b)$ space, which is $\cOtilde(b)$ when $b\geq \sqrt{n}$.
\end{restatable}
\begin{proof}
Let $\textsf{occ}(h)$ be the list of occurrences of anchor $h$ in $T$, let $B=\bigcup_{h\in A}\textsf{occ}(h)$ and $B'=\bigcup_{h\in A}\textsf{occ}(h)+b-1$. 
Recall that $D(h)$ consists of two \emph{compacted} tries $\textsf{TP}^{r}(h)$ and $\textsf{TS}(h)$. We will first construct two global compacted tries $\textsf{TP}^{r}(A)$ and $\textsf{TS}(A)$ for all anchors in $A$, and then extract from them subtries $\textsf{TP}^{r}(h)$ and $\textsf{TS}(h)$ for each $h\in A$. 

$\textsf{TP}^{r}(A)$ and $\textsf{TS}(A)$ are constructed in the same way, except that for $\textsf{TP}^{r}(A)$ we consider the reversal of strings. 
To construct $\textsf{TS}(A)$ we employ Theorem~\ref{the:sparseST} on set $B$, as it is essentially the sparse suffix tree for the suffixes starting at positions in $B$; and to construct $\textsf{TP}^r(A)$ we employ Theorem~\ref{the:sparseST} on the reverse of $T$ and set $B'$.
Once we have constructed $\textsf{TS}(A)$ and $\textsf{TP}^r(A)$, to extract subtries $\textsf{TS}(h)$ and $\textsf{TP}^r(h)$ for $h\in A$ it suffices to spell $h$ from the root of $\textsf{TS}(A)$ (resp. $h^r$ from the root of $\textsf{TP}^r(A)$) and take the subtrie below.

The time and space complexity of computing $\textsf{TS}(A)$ and $\textsf{TP}^r(A)$ is as follows. The size of sets $B$ and $B'$ is $\cO(n/b\cdot\log n)=\cO(b\log n)$ when $b\geq \sqrt{n}$, thus by Theorem~\ref{the:sparseST} we make use of $\cO(b\log n)$ words of space. Again by Theorem~\ref{the:sparseST}, the overall time complexity to construct them is $\tilde{\cO}(n)$. To find the right subtrie for each $h\in A$ we then spend $\cOtilde(b)$ time for each of the $\cO(n/b\log n)$ anchors of $A$, thus again $\cOtilde(n)$ time overall.
\end{proof}
\subparagraph{Computing $S_T(k)$ Using $D(h)$.}~Similar to Section~\ref{sec:improve}, in each phase $\alpha$ we fill in an auxiliary array $I=I[1\dd b]$ such that, at the end of the phase, $I[q]$ contains the number of positions $i$ in $T$ such that the shortest substring that does not occur in $T$ before position $i$ is of length $\alpha b+q$. 
We proceed as follows.
We consider one position of $T$ at the time, from left to right. When we are at position $i$, let $h$ be the leftmost anchor occurring at some position $i'\ge i$. We binary search for the smallest position $j$ such that $T[i\dd j]$ does not occur to the left of $i$ using $D(h)$. 
We first identify in $\textsf{TP}^{r}(h)$ the highest ancestor $u$ of leaf $i'$ with string depth at least $i'-i$. This corresponds to answering a \emph{weighted level ancestor} query~\cite{DBLP:conf/cpm/FarachM96} on $\textsf{TP}^{r}(h)$, where the weight of each node is its depth. After linear-time preprocessing, weighted ancestor queries for nodes of a weighted tree with integer weights from a universe $[1 \dd U]$ can be answered in $\cO(\log \log U)$ time~\cite{DBLP:journals/talg/AmirLLS07}. In our case, the queries thus cost $\cO(\log \log n)$ time.

We then start binary searching for the leftmost position $j$ such that $T[i\dd j]$ does not occur to the left of position $i$ and such that $|T[i\dd j]|\in[\alpha b+1, (\alpha + 1)b]$: we thus look for $j$ in the range $[i+\alpha b,i+(\alpha+1)b-1]$. For each value $j$ considered in the binary search, we find in $\textsf{TS}(h)$ the highest ancestor $v$ of leaf $i'$ with string depth at least $j-(i'+b)$, by answering a weighted level ancestor query.
We then need to check whether $T[i\dd j]$ occurs somewhere to the left of $i$, in correspondence of a previous occurrence of anchor $h$, in which case we increase $j$ in the next step; or it does not occur before, in which case we decrease $j$.
We do so by looking at the leaves (occurrences of $h$) in the subtree below $u$ in $\textsf{TP}^{r}(h)$, denoted by $\textsf{TP}^{r}(h)|u$, and in the subtree below $v$ in $\textsf{TS}(h)$, $\textsf{TS}(h)|v$.
Every leaf in the intersection of the two subsets of leaves corresponds to an occurrence of $T[i\dd j]$ in $T$. The information we need is whether $i'$ is the smallest leaf in the intersection, meaning that $T[i\dd j]$ does not occur anywhere before. This reduces to a 2D range searching problem.

We assume that each leaf of each tree has a unique identifier, independent from their label and such that the identifiers of the leaves of any subtree form a contiguous range. For each leaf $\ell$, its identifiers in $\textsf{TP}^{r}(h)$ and $\textsf{TS}(h)$ give the coordinates of a point on a plane, to which we assign $\ell$ as weight. By construction, the points corresponding to leaves in the intersection of $\textsf{TP}^{r}(h)|u$ and $\textsf{TS}(h)|v$ are contained in a rectangle: we need to find the point with the smallest weight there and check whether it is $i'$ or not.
Such queries can be answered in time $\cO(\log s)$ with a data structure that is constructed in time and space $\cO(s\log s)$, where $s$ is the total number of points~\cite{chazelle1988functional}.
At the end of the binary search, if $j=i+\alpha b+q$ we increase the counter at $I[q]$ by one, unless $j=i+(\alpha+1)b-1$ and $T[i\dd j]$ occurs before $i$, in which case we do not increase any counters. We finally move to the next position of $T$.

\begin{lemma}\label{lem:baperiodic}
Assume that all length-$b$ substrings of $T$ are not strongly periodic.
Then $\delta$ can be computed in $\tilde{\cO}(n^2/b)$ time using $\tilde{\cO}(n/b+b)$ space, which is $\tilde{\cO}(b)$ when $b\geq \sqrt{n}$.
\end{lemma}

\begin{proof}
Set $A$ is selected in $\tilde{\cO}(n^2/b)$ time and $\tilde{\cO}(n/b+b)$ space as per Lemma~\ref{lem:anchorexist}, and $D(h)$ can be computed in the same time and space for all $h\in A$ and all phases, as per Lemma~\ref{lem:D(h)}. In each phase $\alpha$, we go over the $n$ positions of $T$ one at a time. At each position $i$ we binary search for the shortest substring not occurring before $i$ in $\cO(\log \alpha b)$ steps, each requiring $\cO(\log n)$
time. Over all $\cO(n/b)$ phases, this requires $\tilde{\cO}(n^2/b)$ time and $\tilde{\cO}(n/b)$ space.
\end{proof}

\subsection{Processing the $b$-Runs}\label{sub:bperiodic}
Recall that we have computed, as per Lemma~\ref{lem:regions}, a representation of all the $b$-runs of $T$. In this section, we only focus on the substrings of length at least $b$ and periods at most $b/4$. Every occurrence of such a substring is fully contained in some $b$-run, and for ease of presentation we will assume that in phase $\alpha$, in which we process substrings of length $k\in[\alpha b+1,(\alpha+1)b]$, every $b$-run is longer than $\alpha b$. Observe that each substring of a $b$-run $T[s\dd e]$ with root $t$ occurs also as a prefix of some fragment starting within the first $|t|$ positions of the run, which we call its \emph{relevant} range. Since we aim to identify the leftmost occurrence of each substring of $T$, we can ignore all positions of a $b$-run after its relevant range. By slightly abusing notation, we select as anchors some fragments of the $b$-runs of $T$, instead of selecting substrings together with the whole set of their occurrences. However, this set of anchors must have the following property, that for the anchors of Section~\ref{sub:baperiodic} held naturally: for any two occurrences of the same substring in the relevant ranges, the leftmost occurrence of any anchor therein is at the same offset from the beginning of the substring. In phase $\alpha$ we use as anchors the first two occurrences of the root in each $b$-run: let $H$ be this set of fragments of $T$. Clearly, $H$ is of size $\cO(n/b)$ because the representation of all the $b$-runs is of such size.
\begin{restatable}{lemma}{lemperiodicanchors}\label{lem:periodicanchors}
For any two occurrences of the same substring of length at least $b$ and period at most $b/4$, both starting in the relevant ranges of the $b$-runs of $T$, the leftmost occurrence of any $h\in H$ in each of them is at the same offset from the beginning of the substring.
\end{restatable}

\begin{proof}
Let $Y=t[d\dd |t|]t^{\beta}t[1\dd f]$ be a fragment occurring at the first $t$ positions in some $b$-run $T[s\dd e]=t[q\dd |t|]t^{\gamma}t[1\dd g]$. The anchors within $T[s\dd e]$ are, by definition, the occurrences of $t$ at position $p_1=s+(|t|-q+1)\mod |t|$ and $p_2=p_1+|t|$. If $d\geq q$, the leftmost occurrence of any anchors in $Y$ is at $p_1$, which is at offset $|t|-d+2$ in $Y$. Otherwise, if $d<q$, the leftmost occurrence of any anchors in $Y$ is at $p_2$, which is in any case at offset $|t|-d+2$ in $Y$.

Consider another occurrence of $Y$ in the first $|t|$ positions of some other $b$-run $T[s'\dd e']=t[q'\dd |t|]t^{\gamma'}t[1\dd g']$. The anchors are the occurrences of $t$ at position $p'_1=s'+(|t|-q'+1)\mod |t|$ and $p'_2=p'_1+|t|$; depending on whether $d\geq q'$ or not, the leftmost occurrence of any anchor in this occurrence of $Y$ is either $p'_1$ or $p'_2$, in either case at offset $|t|-d+2$ in $Y$.
\end{proof}
Let $P$ be the set of roots of the $b$-runs of $T$.
We construct a data structure $D(P)$ for all roots $t\in P$ similar to what we do in Section~\ref{sub:baperiodic}, but we use only the occurrences of $t$ corresponding to fragments in $H$. We then proceed as described in Section~\ref{sub:baperiodic} to fill in array $I$, except that, in each $b$-run with root $t$, we disregard any position after the first $|t|$. 

We have arrived at the following lemma.

\begin{lemma}\label{lem:processbruns}
The substrings of $T$ that are fully contained within a $b$-run can be processed in $\cOtilde(n^2/b)$ time using $\cO(n/b)$ space, which is $\cO(b)$ when $b\geq \sqrt{n}$.
\end{lemma}

\subsection{Computing $S_T(k)$ for Large $k$}\label{sec:mainlargek}
The occurrences of anchors $h\in A$ selected for the $b$-gaps anchor all the fragments
fully contained in a $b$-gap and possibly some other fragments. However, we are not guaranteed
that this holds for any fragment not fully contained in a $b$-run.
Consider a fragment $T[i\dd j]$ of length at least $b$ with period larger than $b/4$ (thus, not contained
in any $b$-run) but containing a strongly periodic length-$b$ fragment $T[i'\dd j']$ inside (so, not contained in any $b$-gap).
Then, $T[i'\dd j']$ is fully contained in some $b$-run $T[s\dd e]$. Because $T[i\dd j]$ is not fully contained
in $T[s \dd e]$, either $T[s-1 \dd s+b-2]$ or $T[e-b+2\dd e+1]$ (that is, a length-$b$ substring with exactly one letter
before or after the $b$-run) is fully within $T[i\dd j]$. 
This suggest that we should augment $A$ with the following length-$b$ substrings: for each $b$-run $T[s\dd e]$, $T[s-1\dd s+b-2]\in A$ and $T[e-b+2\dd e+1]\in A$, and we consider all their occurrences in $T$.
By the above reasoning, this guarantees that $T[i\dd j]$ contains an occurrence of some anchor inside.
We are defining only $\cO(n/b)$ new anchors, but then we need to consider all of their occurrences.
Therefore, we need to argue that the total number of occurrences of the new anchors is $\cO(n/b)$.
It is enough to show this for the occurrences of the anchors $T[s-1 \dd s+b-2]$, where the period of $T[s \dd s+b-2]$ is at most $b/4$.
We claim that for any two such occurrences $T[s-1 \dd s+b-2]$ and $T[s'-1 \dd s'+b-2]$ with $s<s'$
we have $s+b/2 < s'$: otherwise $T[s \dd s+b-2]$ and $T[s' \dd s'+b-2]$ overlap by at least $b/2$ positions,
but two $b$-runs cannot overlap by $b/2$ positions, a contradiction.
We generate all these occurrences and then process all the anchors as in Section~\ref{sub:baperiodic}. 
The only difference is the starting range associated with the anchors obtained from the suffix of some $b$-run: when they are not preceded by another anchor within $cb$ positions, we take as starting range the $\alpha b$ positions preceding the anchor. 

Let us put everything together.
Before computing $S_T(k)$ in phases, we identify the $b$-runs and the $b$-gaps of $T$ as per Lemma~\ref{lem:regions}. We then extract a set of anchors from the $b$-gaps as described in Lemma~\ref{lem:anchorexist}, and we complement it with the length-$b$ substrings that start one position before each $b$-run, and with the length-$b$ substrings that end one position after the end of each $b$-run, to complete the set $A$ of anchors. We then compute the list of occurrences of each $h\in A$; we also identify the relevant ranges within each $b$-run.
We then proceed in phases.
In each phase, we scan $T$ from left to right and process all positions in $b$-gaps as per Section~\ref{sub:baperiodic}. All positions within a $b$-run are processed as per Section~\ref{sub:bperiodic}, and additionally as per Section~\ref{sub:baperiodic}, when they are within the starting range of an occurrence of some $h\in A$. At the end of a phase $\alpha$, we have computed an auxiliary array $I$ such that $I[h]$ gives the number of positions $i$ of $T$ such that the shortest substring that does not occur in $T$ before position $i$ is of length $\alpha b+h$. We use $I$ to compute $S_T(k)$ for each $k\in[\alpha b+1, (\alpha+1)b]$ as in Section~\ref{sec:improve}.

By combining Lemmas~\ref{lem:baperiodic} and~\ref{lem:processbruns} we arrive at Theorem~\ref{the:n2}, the main result of this paper.

\section{Substring Complexity from the Combinatorial Point of View}\label{sec:combinatorial} 
%
Knowing the substring complexity of a string can also be used to find other regularities. To mention a few, we have the following straightforward implications in sublinear working space: 
\begin{itemize}
    \item $T$ has a substring of length $k$ repeating in $T$ if and only if $S_T(k)<n-k+1$.
    This yields the length $r$ of the longest repeated substring of $T$ (also known as the repetition index of $T$)~\cite{DBLP:conf/focs/Weiner73}. 
     It is worth noticing that $S_T(k+1)=S_T(k)-1$ for every $k> r$~\cite{DBLP:journals/tcs/Luca99} and that $r$ approximates $\cO(\log_{|\Sigma|} n)$ when $T$ is randomly generated by a memoryless source~\cite{DBLP:journals/tcs/FiciMRS06}.
    \item A string $S$ is called a {\em minimal absent word} of $T$ if $S$ does not occur in $T$ but all proper substrings of $S$ occur in $T$. The length $\ell$ of a longest minimal absent word of $T$ is equal to $2+r$~\cite{DBLP:journals/tcs/FiciMRS06}. This quantity is important because if two strings $X$ and $Y$ have the same set of distinct substrings up to length $\ell$, then $X=Y$~\cite{DBLP:journals/tcs/FiciMRS06,DBLP:journals/tcs/CarpiL01}. The length of a shortest absent word~\cite{DBLP:journals/ipl/WuJS10} of $T$ over alphabet $\Sigma$ is equal to the smallest $k$ such that $S_T(k)<|\Sigma|^k$.
\item The longest common substring of strings $X$ and $Y$ is equal to the largest $k$ such that $S_X(k)+S_Y(k)>S_{X\#Y}(k)-k$, where $\#$ does not occur in $X$ nor in $Y$, since there are precisely $k$ distinct substrings of length $k$ containing the letter $\#$ in $X\#Y$.
\end{itemize}

The substring complexity function is well studied in the area of combinatorics on words, both for finite and infinite strings. 
However, the normalization $S(k)/k$ and its supremum $\delta$ have not been considered until very recently. In~\cite{DBLP:journals/tcs/Luca99} it is proved that the substring complexity $S_T(k)$ of a string $T$ takes its maximum precisely for $k=R$, where $R$ is the minimum length for which no substring of $T$ has occurrences followed by different letters, and one has $S_T(R)=n+1-\max\{R,K\}$, where $K$ is the length of the shortest unrepeated suffix of $T$. 
But this seems to be of little help in understanding the behaviour of the \emph{normalized} substring complexity  $S_T(k)/k$.

\section{Approximating $\delta$ in Sublinear Space}\label{sec:finale}

Our algorithms compute the \emph{exact value} of $\delta$. If one is interested in a constant-factor approximation of $\delta$ (e.g., an algorithm's complexity has a polynomial dependency on $\delta$~\cite{DBLP:journals/tit/KociumakaNP23}), then there is a simple algorithm in our model based on the following combinatorial observation, which follows directly by the number of fragments of length $\ell$ of a string of length $n$ being $n-\ell+1$, and by the fact that each fragment of length $\ell'>\ell$ has a prefix of length $\ell$.

\begin{observation}\label{obs:distinct}
    For any string $T$, let $S_T(k)$ be the number of distinct substrings of length $k$. The number $S_T(k')$ of distinct substrings of any length $k'>k$ is at least $S_T(k)-(k'-k)$.
\end{observation}

\begin{lemma}\label{lem:delta'}
    Let $\delta'=\sup\{\frac{S_T(2^d)}{2^d}\mid d=0,\ldots,\log n\}$. Then $\delta \le 2\delta'+1$.
\end{lemma}

\begin{proof}
    Let $\delta=\frac{S_T(k)}{k}$ for some $k\geq 1$, and let $d$ be the integer such that
    \begin{equation}\label{eq:powers}
        2^{d}\leq k < 2^{d+1}.
    \end{equation}
  By the definition of $\delta'$, we have that
    $\delta'\geq \frac{S_T(2^{d+1})}{2^{d+1}}$.
    By applying Observation~\ref{obs:distinct}, we obtain:
    \begin{equation*}\label{eq:proof}
        \frac{S_T(2^{d+1})}{2^{d+1}}\ge \frac{S_T(k)-(2^{d+1}-k)}{2^{d+1}}\ge \frac{S_T(k)-(2^{d+1}-2^{d})}{2^{d+1}}\ge \frac{S_T(k)}{2k}-\frac{2^d}{2^{d+1}}=\frac{1}{2}\delta-\frac{1}{2}.
    \end{equation*}
\end{proof}

Recall that the algorithm underlying
Theorem~\ref{the:n3} works in $\frac{n}{b}$ phases, where each phase handles a range of $b$ lengths $k$. By plugging in Lemma~\ref{lem:delta'}, the number of phases become $\Theta(\log n)$ -- instead of $\Theta(n/b)$ -- and so we obtain a simple
$\cOtilde(n^2/b)$-time and $\cO(b)$-space algorithm to approximate $\delta$, within a constant factor, in the comparison model.

\bibliography{bibliography}

\end{document}